\documentclass[11pt,letterpaper]{article}
\pdfoutput=1 

\usepackage[linesnumbered,boxruled,vlined]{algorithm2e}

\usepackage[utf8]{inputenc}
\usepackage[T1]{fontenc}

\usepackage{cite} 
\usepackage{graphicx}
\usepackage{threeparttable}
\usepackage{booktabs}
\usepackage[margin=1in]{geometry}
\usepackage{amsmath,setspace,amssymb} 
\usepackage{amsthm} 
\usepackage{latexsym}
\usepackage{comment}
\usepackage{enumitem}
\setlist{itemsep=2pt, topsep=4pt, parsep=0pt} 
\setlist[enumerate]{label=(\arabic*)} 
\usepackage[unicode,pdfencoding=auto,colorlinks=true,linkcolor=magenta,citecolor=blue]{hyperref}
\DeclareUrlCommand\path{\urlstyle{same}} 
\usepackage{bookmark}
\usepackage{color}
\usepackage{thm-restate}
\usepackage{here}        
\usepackage[capitalise,noabbrev,nameinlink]{cleveref} 

\usepackage{libertinus}

\graphicspath{{../figures/}{./}}

\theoremstyle{plain}
\newtheorem{theorem}{Theorem}
\newtheorem{lemma}[theorem]{Lemma}
\newtheorem{corollary}[theorem]{Corollary}

\theoremstyle{definition}
\newtheorem{definition}[theorem]{Definition}

\usepackage[most]{tcolorbox}
\newtcolorbox{problemdef}{
  colback=white,
  colframe=black,
  boxrule=0.8pt,
  arc=0pt,
  outer arc=0pt,
  left=4pt,
  right=4pt,
  top=2pt,
  bottom=2pt,
  before skip=6pt,
  after skip=6pt
}

\newcommand{\dist}{\mathsf{dist}}

\newcommand{\indeg}{\mathrm{indeg}}

\newcommand{\tabnote}[1]{{\fontsize{8pt}{9.6pt}\selectfont #1}} 

\allowdisplaybreaks[2]

\title{A Simple Algorithm for the Directed Multiple Source Replacement Paths Problem}
\author{
Kaito Harada \\
The University of Osaka\\
k-harada@ist.osaka-u.ac.jp \and 
Taisuke Izumi \\
The University of Osaka\\
t-izumi@ist.osaka-u.ac.jp
}

\date{}

\begin{document}

\maketitle
\thispagestyle{empty}


\begin{abstract}
    In the \emph{replacement paths} (RP) problem, we are given a graph $G = (V, E)$ with $n = |V|$ and $m = |E|$, together with two vertices $s, t \in V$, and are asked to compute the shortest-path distance from $s$ to $t$ in $G \setminus e$ for every failed edge $e \in E$.
    The \emph{multiple source replacement paths} (MSRP) problem is its natural generalization: given a set $S \subseteq V$ of $\sigma$ sources, compute the replacement path distances for all pairs in $S \times V$.

    In this paper, we present a randomized combinatorial algorithm that solves MSRP on unweighted directed graphs in $\tilde{O}(m\sqrt{\sigma n} + \sigma n^2)$ time\footnote{The $\tilde{O}(\cdot)$ notation omits polylogarithmic
    factors, i.e., $\tilde{O}(f(n)) = O(f(n)\,\mathrm{polylog}(n))$, and the notation $\tilde{\Theta}(\cdot)$ is defined analogously.}, with all the output distances correct with high probability.
        This improves the best known bound $\tilde{O}(m\min\{\sigma\sqrt{n}, n\} + \sigma n^2)$ for directed graphs, which is obtained either by running the single source RP algorithm of Chechik and Magen [ICALP'20] from each source separately or by constructing and querying the all-pairs distance sensitivity oracle of Bernstein and Karger [STOC'09].
        Our running time is essentially tight among combinatorial algorithms because Gupta, Jain, and Modi [PODC'20] proved a lower bound of $m{(\sigma n)}^{1/2-o(1)}$ for such algorithms, which holds even on undirected graphs, and the additive term $\sigma n^2$ is proportional to the time needed to write down the $\Theta(\sigma n^2)$ output distances.

        The algorithm is also remarkably simple.
        Its key technical idea is inspired by the auxiliary graph construction of Gupta, Jain, and Modi: for each source $s$, one builds a weighted graph whose vertices include one for every pair $(t, e)$ of a terminal $t$ and a failed edge $e$ on a shortest $s$--$t$ path, and in which the shortest-path distance from a single super source to the vertex of $(t, e)$ equals the replacement path distance for $(t, e)$ in the original graph.
        In their algorithm, however, this direct reduction to a shortest-path computation in the auxiliary graph applies to only one of the three cases distinguished by their analysis, and the other two are handled by separate subroutines and analyses that rely on a landmark sampling technique.
        We show that a slight modification of the construction, which embeds the information derived from the landmark sampling into the edges of the graph itself, extends this reduction to \emph{all} the cases at once.
        Hence, once this extended auxiliary graph is constructed, all the replacement path distances from $s$ are obtained by a single run of Dijkstra's algorithm on it. In addition,
    this modification bypasses the part of their analysis crucially relying on the graph being undirected, which is exactly what makes the directed case accessible.
    This simplicity not only makes the description of the algorithm concise, but also substantially shortens its analysis.
\end{abstract}

\section{Introduction}
\subsection{Background and Our Result}\label{subsec:result}

The \emph{replacement paths problem} (RP) is defined as follows: given a graph $G = (V, E)$ with $n = |V|$ and $m = |E|$, and two vertices $s, t \in V$, for each edge $e$ on a shortest $s$--$t$ path $P$ in $G$, compute the shortest-path distance from $s$ to $t$ (or a shortest path itself) in $G \setminus e$, the graph obtained from $G$ by removing $e$.
This problem abstracts the computation of detour routes under a link failure in networks and also arises in classical applications such as the computation of Vickrey prices in shortest path auctions~\cite{HS01} and the finding of $k$ simple shortest paths~\cite{Yen71,KIM82,RZ12}.

The most naive algorithm for RP first computes a shortest $s$--$t$ path $P$, and then, for each edge $e$ on $P$, runs a single-source shortest path algorithm on $G \setminus e$ (i.e., breadth-first search for unweighted graphs or Dijkstra's algorithm for non-negatively weighted graphs).
Since the number of edges on $P$ is at most $n-1$, this approach runs in $\tilde{O}(mn)$ time.
A central question in the study of RP is whether and to what extent this baseline algorithm can be improved.

There are two major algorithmic approaches for RP, each forming a distinct line of research.
\emph{Combinatorial algorithms} are those that do not rely on fast matrix multiplication.
The best known combinatorial algorithms run in $\tilde{O}(m)$ time for non-negatively weighted undirected graphs~\cite{MMG89,HS01} and in $\tilde{O}(m \sqrt{n})$ time for unweighted directed graphs~\cite{RZ12,ACC19}.
\emph{Algebraic algorithms} exploit fast matrix multiplication.
For directed graphs with integer edge weights in $[-M, M]$, RP can be solved in $\tilde{O}(Mn^{\omega})$ time~\cite{VW11,CN20}, where $\omega \in [2, 2.371339)$ denotes the matrix multiplication exponent~\cite{ADWXXZ25}.\footnote{A recent preprint by Dupont et al.~\cite{DEKMRSZAVB26} improves this bound to $\omega < 2.371177$.}
Conditional lower bounds have also been investigated extensively.
For directed graphs with arbitrary edge weights, RP is equivalent to the all-pairs shortest paths (APSP) problem under subcubic reductions~\cite{VWW18}.\footnote{An APSP algorithm solves RP with $\tilde{O}(n^2)$ overhead, and an $O(n^{3-\varepsilon})$-time algorithm for RP yields an $\tilde{O}(n^{3-\varepsilon/3})$-time algorithm for APSP~\cite{VWW18}. Very recently, APSP on directed graphs with polynomially bounded integer weights was shown to be solvable in $O(n^{2.9995})$ time~\cite{AV26}, and hence so is RP; the algorithm is algebraic, and no truly subcubic combinatorial algorithm is known even for Boolean matrix multiplication.}
For unweighted directed graphs, any combinatorial algorithm requires $mn^{1/2-o(1)}$ time under the combinatorial BMM conjecture~\cite{VWW18}, which states that any combinatorial algorithm for the Boolean matrix multiplication of two $n \times n$ matrices requires $n^{3-o(1)}$ time.
This bound indicates that the $\tilde{O}(m \sqrt{n})$ upper bound of combinatorial algorithms is essentially tight.

A natural generalization of RP, which has been actively studied in recent years, is to relax the number of sources and terminals.
In the \emph{single source replacement paths} (SSRP) problem, we are given a single source $s$, and are asked to compute the replacement path distances from $s$ to all terminals $t \in V$ for all edge failures.
For unweighted graphs, SSRP can be solved in $\tilde{O}(m\sqrt{n} + n^2)$ time on both undirected~\cite{CC19,BCFS21} and directed~\cite{CM20,BCCFS22} graphs.
Generalizing the problem further, Gupta, Jain, and Modi~\cite{GJM20} introduced the \emph{multiple source replacement paths} (MSRP) problem, which is formalized as follows.
\begin{problemdef}
    \textbf{Multiple Source Replacement Paths Problem (MSRP)}

    \textbf{Input:}
    A graph $G = (V, E)$ and a set $S \subseteq V$ of $\sigma$ sources.

    \textbf{Output:}
    The distance $\dist_{G \setminus e}(s, t)$ from $s$ to $t$ in $G \setminus e$, for each $(s, t) \in S \times V$ and edge $e \in E$.
\end{problemdef}

We remark that the output size is actually $\Theta(\sigma n^2)$ because for each source $s \in S$ and terminal $t \in V$, we only care about edges on a shortest $s$--$t$ path.

MSRP contains the problems mentioned above as special cases: taking $S = \{s\}$ recovers SSRP, and taking $S = V$ yields the all-pairs setting.
For undirected unweighted graphs, Gupta, Jain, and Modi~\cite{GJM20} gave a randomized combinatorial algorithm that solves MSRP in $\tilde{O}(m\sqrt{\sigma n} + \sigma n^2)$ time.
They also proved a conditional lower bound of $m{(\sigma n)}^{1/2-o(1)}$ for combinatorial algorithms based on the combinatorial BMM conjecture.
While this lower bound is stated for undirected graphs, it holds for the directed case as well: the graph constructed in their reduction is layered, and orienting all of its edges in the direction from the sources toward the terminals leaves the argument intact.
For directed graphs, however, no algorithm designed specifically for MSRP has been known, and there are two generic combinatorial baselines: running the directed SSRP algorithm from each source separately takes $\tilde{O}(\sigma m\sqrt{n} + \sigma n^2)$ time, and constructing the all-pairs distance sensitivity oracle of Bernstein and Karger~\cite{BK09} and querying all the $\Theta(\sigma n^2)$ relevant triples takes $\tilde{O}(mn + \sigma n^2)$ time.
Thus the best known combinatorial bound has been $\tilde{O}(m\min\{\sigma\sqrt{n}, n\} + \sigma n^2)$.
\cref{tab:known-results,tab:algebraic-results} summarize the known combinatorial and algebraic results for the replacement path problems, respectively.

In this paper, we close the gap between the undirected and directed bounds: we improve the first term of the directed bound to $\tilde{O}(m\sqrt{\sigma n})$, an improvement by a factor of $\min\{\sqrt{\sigma}, \sqrt{n/\sigma}\}$, and give a near-optimal interpolation between the single source ($\sigma = 1$) and all-pairs ($\sigma = n$) endpoints.

\begin{table}[t]
    \centering
    \caption{Known combinatorial results for the replacement paths problems. The lower bounds marked with $\ddagger$ are conditional on the combinatorial BMM conjecture. The bound marked with $\dagger$ is obtained by querying an all-pairs distance sensitivity oracle with all the $\Theta(\sigma n^2)$ relevant triples.}\label{tab:known-results}
    \setlength{\tabcolsep}{4pt}
    \begin{tabular}{cccccc}
        \toprule
        $(\#s, \#t)$  & Direction   & Weights    & Det./Rand. & Upper Bound                                                            & Lower Bound                                       \\
        \midrule
        $(1, 1)$      & undir.      & unwei.     & det.       & $O(m+n)$~\tabnote{\cite{LL14}}                                           & -                                                 \\
        $(1, 1)$      & dir.        & unwei.     & det.       & $\tilde{O}(m\sqrt{n})$~\tabnote{\cite{RZ12,ACC19}}                     & ${mn^{1/2-o(1)}}^{\ddagger}$~\tabnote{\cite{VWW18}}            \\
        $(1, 1)$      & undir.      & $[1, n^c]$ & det.       & $\tilde{O}(m)$~\tabnote{\cite{MMG89,HS01}}                             & -                                                 \\
        $(1, 1)$      & dir.        & $[1, n^c]$ & det.       & APSP-time~\tabnote{\cite{VWW18}}                                                     & -                  \\
        \midrule
        $(1, n)$      & undir.      & unwei.     & det.       & $\tilde{O}(m\sqrt{n} + n^2)$~\tabnote{\cite{CC19,BCFS21}}              & ${mn^{1/2-o(1)}}^{\ddagger}$~\tabnote{\cite{CC19}}             \\
        $(1, n)$      & dir.        & unwei.     & det.       & $\tilde{O}(m\sqrt{n} + n^2)$~\tabnote{\cite{CM20,BCCFS22}}             & ${mn^{1/2-o(1)}}^{\ddagger}$~\tabnote{\cite{CC19}}             \\
        $(1, n)$      & undir./dir. & $[1, n^c]$ & rand.       & APSP-time~\tabnote{\cite{GV20}}                                   & APSP-time~\tabnote{\cite{VWW18,CC19}}                   \\
        \midrule
        $(\sigma, n)$ & undir.      & unwei.     & rand.      & $\tilde{O}(m\sqrt{\sigma n} + \sigma n^2)$~\tabnote{\cite{GJM20}}      & ${m{(\sigma n)}^{1/2-o(1)}}^{\ddagger}$~\tabnote{\cite{GJM20}} \\
        $(\sigma, n)$ & dir.        & unwei.     & rand.      & $\tilde{O}(m\sqrt{\sigma n} + \sigma n^2)$~\tabnote{[\cref{thm:main}]} & ${m{(\sigma n)}^{1/2-o(1)}}^{\ddagger}$~\tabnote{\cite{GJM20}} \\
        $(\sigma, n)$ & dir.        & $[1, n^c]$ & det.       & $\tilde{O}(mn + \sigma n^2)^{\dagger}$~\tabnote{\cite{BK09}}           & APSP-time~\tabnote{\cite{VWW18,CC19}}                  \\
        \bottomrule
    \end{tabular}
\end{table}

\begin{table}[t]
    \centering
    \caption{Known algebraic algorithms for the replacement paths problems. The bounds marked with $\dagger$ are obtained by querying all-pairs distance sensitivity oracles with all the $\Theta(\sigma n^2)$ relevant triples.}\label{tab:algebraic-results}
    \setlength{\tabcolsep}{4pt}
    \begin{tabular}{ccccc}
        \toprule
        $(\#s, \#t)$  & Direction & Weights   & Det./Rand. & Upper Bound                                                           \\
        \midrule
        $(1, 1)$      & dir.      & $[-M, M]$ & det.       & $\tilde{O}(Mn^{\omega})$~\tabnote{\cite{VW11,CN20}}                   \\
        $(1, 1)$      & dir.      & $[-M, M]$ & det.       & $\tilde{O}(M^{0.681}n^{2.575})$~\tabnote{\cite{VW11}}                 \\
        \midrule
        $(1, n)$      & undir.    & $[1, M]$  & det.       & $\tilde{O}(Mn^{\omega})$~\tabnote{\cite{GV20,BCFS21}}                 \\
        $(1, n)$      & dir.      & $[1, M]$  & rand.      & $\tilde{O}(Mn^{\omega})$~\tabnote{\cite{GV20}}                        \\
        $(1, n)$      & dir.      & $[-M, M]$ & rand.      & $O(M^{0.7519}n^{2.5286})$~\tabnote{\cite{GV20}}                       \\
        $(1, n)$      & dir.      & $[-M, M]$ & rand.      & $O(M^{0.8043}n^{2.4957})$~\tabnote{\cite{GPVX21}}                     \\
        $(1, n)$      & dir.      & $[-M, M]$ & rand.      & $\tilde{O}(M^{0.8825}n^{2.4466})$~\tabnote{\cite{D23}}                \\
        \midrule
        $(\sigma, n)$ & dir.      & unwei.    & rand.      & $O(n^{2.529} + \sigma n^2)^{\dagger}$~\tabnote{\cite{KS23}}           \\
        $(\sigma, n)$ & dir.      & $[1, M]$  & det.       & $O(Mn^{2.8068} + \sigma n^2)^{\dagger}$~\tabnote{\cite{BCCFS22}}      \\
        $(\sigma, n)$ & dir.      & $[1, M]$  & rand.      & $O(Mn^{2.5794} + \sigma n^2)^{\dagger}$~\tabnote{\cite{GR21}}         \\
        $(\sigma, n)$ & dir.      & $[-M, M]$ & rand.      & $\tilde{O}(Mn^{2.8729} + \sigma n^2)^{\dagger}$~\tabnote{\cite{CC20}} \\
        \bottomrule
    \end{tabular}
\end{table}

\begin{theorem}\label{thm:main}
    There is a randomized combinatorial algorithm for the MSRP problem on unweighted directed graphs that runs in $\tilde{O}(m\sqrt{\sigma n} + \sigma n^2)$ time.
    Our randomized algorithm is Monte Carlo with a one-sided error: it always outputs distances that are at least the exact distances, and all the outputs are exact with high probability.
\end{theorem}

Since the output consists of $\Theta(\sigma n^2)$ distances in the worst case, the first term of our running time matches the conditional lower bound of Gupta et al.~\cite{GJM20}, and the second term is proportional to the time needed to write down the output.
Hence, our algorithm is near-optimal among combinatorial algorithms.

Beyond the improved running time, a notable feature of our algorithm is its simplicity.
Its key technical idea is inspired by the auxiliary graph construction of Gupta et al.~\cite{GJM20}: for each source $s \in S$, one builds a weighted graph whose vertices represent the states $(t, e)$, i.e., the pairs of a terminal and a failed edge, so that all the replacement path distances from $s$ are obtained by a single run of Dijkstra's algorithm on it.
Their construction covers only one of the three cases distinguished by their analysis, and we show that a slight modification of it, which embeds the information derived from the landmark sampling into the edges of the graph itself, covers all the cases at once.
In addition, this modification bypasses the part of their analysis that relies on the graph being undirected.
The only randomness lies in the standard sampling of landmark vertices, with sampling probabilities decreasing geometrically over the levels, and the analysis requires only elementary arguments.
In particular, our algorithm specialized to $\sigma = 1$ serves as a simple alternative to the known randomized algorithms for the SSRP problem~\cite{CC19,CM20}.

\subsection{Related Work}
The literature on replacement paths and fault-tolerant shortest paths is vast.
The exact algorithms most closely related to our result have already been discussed in the previous subsection; here we briefly review three other directions.

\paragraph{Distance sensitivity oracles.}
A research direction complementary to replacement paths algorithms is to preprocess a graph into a compact data structure called a \emph{distance sensitivity oracle} (DSO) that returns $\dist_{G \setminus e}(s, t)$ when queried with a triple $(s, t, e)$.
For directed weighted graphs, a DSO of nearly optimal size $\tilde{O}(n^2)$ with constant query time can be constructed in $\tilde{O}(mn)$ time~\cite{DTCR08,BK09}, and a line of algebraic constructions has brought the preprocessing time down to subcubic for small integer weights~\cite{CC20,R20,GR21,BCCFS22}.
For unweighted directed graphs, Karczmarz and Sankowski~\cite{KS23} constructed a randomized DSO with constant query time in $O(n^{2.529})$ preprocessing time, which matches the best known running time for the all-pairs shortest paths problem in the same setting.
Closest to our problem is the data structure counterpart of MSRP: oracles that answer queries only from a designated source set $S$.
For undirected unweighted graphs, exact single-failure distance oracles for a source set were given by Bil{\`o}, Choudhary, Gual{\`a}, Leucci, Parter, and Proietti~\cite{BCGLPP18} and by Gupta and Singh~\cite{GS18}, the latter achieving size $\tilde{O}(\sqrt{\sigma} n^{3/2})$ with $\tilde{O}(1)$ query time.
In the single source case, oracles of near-optimal size with $\tilde{O}(1)$ query time are known~\cite{BCFS21,DG22}; notably, the oracle of Dey and Gupta~\cite{DG22} can be built in $\tilde{O}(m\sqrt{n})$ time, which matches the conditional lower bound for combinatorial SSRP algorithms.
Recently, Dey and Kavitha~\cite{DK25} constructed fault-tolerant \emph{approximate} distance oracles for a source set with constant stretch and constant query time on undirected weighted graphs.

\paragraph{Approximate settings.}
Another way to bypass the conditional lower bounds is to allow approximation.
Bernstein~\cite{B10} showed that $(1+\varepsilon)$-approximate replacement paths on directed weighted graphs can be computed in $\tilde{O}(m/\varepsilon)$ time, nearly matching the time needed to read the input.
In the single source setting, data structures of near-optimal size that report $(1+\varepsilon)$-approximate distances under a single failure have been developed~\cite{BK13,BCHR20}, and such an oracle can be constructed in nearly linear time even on directed weighted graphs~\cite{HKIM24}.
These results show that allowing a $(1+\varepsilon)$ stretch makes replacement path computation drastically easier.

\paragraph{Multiple failures.}
Extending fault tolerance from a single failed edge to $f \ge 2$ failures is another active line of research.
On the data structure side, exact and approximate oracles supporting multiple failures have been developed~\cite{DP09,WY13,BS19,DR22,CCFK17,DGR21}, culminating in an oracle of Dey and Gupta~\cite{DG24} that is near optimal in both size and query time for every constant number of failures on undirected weighted graphs.
On the algorithmic side, Vassilevska Williams, Woldeghebriel, and Xu~\cite{VWX22} initiated the systematic study of the replacement paths problem under two edge failures, and Chechik and Zhang gave a near-optimal $(1+\varepsilon)$-approximation algorithm for weighted directed graphs~\cite{CZ24a} and the first truly subcubic exact combinatorial algorithm for unweighted directed graphs~\cite{CZ24b}.
Very recently, Nogler and Vassilevska Williams~\cite{NVW26} showed that on undirected graphs, the two-failure replacement paths problem reduces to SSRP, yielding improved algorithms in several settings together with matching conditional lower bounds.
For three failures, Chi, Duan, Wang, and Xie~\cite{CDWX25} presented a deterministic $\tilde{O}(n^3)$-time algorithm on undirected weighted graphs, which almost matches the output size and implies a nearly optimal $\tilde{O}(n^f)$-time algorithm for every $f \ge 3$.

\subsection{Technical Outline}\label{subsec:techniques}

\paragraph{Outline of Algorithm \textsf{GJM}.}
Since our algorithm is based on the undirected MSRP algorithm of Gupta, Jain, and Modi~\cite{GJM20}, which we denote by \textsf{GJM} hereafter, we begin the technical outline of our algorithm by reviewing \textsf{GJM}.
Before reviewing the previous algorithm and ours, we fix the terminology used throughout the paper.
Fix a source $s \in S$, and let $\pi(s, t)$ denote a shortest $s$--$t$ path in $G$ chosen in advance for each terminal $t \in V$.
We call a pair $(t, e)$ of a terminal $t$ and a failed edge $e$ on $\pi(s, t)$ a \emph{state}, and call the replacement path distance $\dist_{G \setminus e}(s, t)$ the \emph{value} of the state $(t, e)$.
That is, the states are exactly the entries of the output for the source $s$, and the task is to compute the value of every state.
First, \textsf{GJM} randomly samples a set $\mathcal{L} = \bigcup_{k} \mathcal{L}_k$ of \emph{multiscale landmarks}, where each $\mathcal{L}_k$ is sampled from all vertices with probability $\frac{1}{2^k}\sqrt{\sigma/n}$ independently. Then, as a preprocessing step, it computes the distance from each landmark to every vertex, by running a BFS starting from each landmark.

We refer to the suffix of the replacement path after departing $\pi(s, t)$ as its \emph{detour part}. By definition, its length is at least as long as $\dist(v, t)$.
The computation of the value of a state $(t, e)$ with $e = (u, v)$ is roughly divided into three cases, according to the scale of $\dist(v, t)$: writing $\Delta_k = \tilde{\Theta}(2^k \sqrt{n/\sigma})$ for the $k$-th scale, a state is called \emph{$k$-far} if $\dist(v, t)$ falls in the $k$-th scale, i.e., $\Delta_k \le \dist(v, t) < \Delta_{k+1}$, and \emph{near} if it falls below the smallest one.
\textsf{GJM} splits the states into the following three cases and runs a different subroutine/analysis on
each of them.
\begin{enumerate}
    \item \emph{$k$-far states.}
          Since $\dist(v, t) \ge \Delta_k$, the detour part is long enough to be hit by $\mathcal{L}_k$, whose density is matched to $\Delta_k$; with high probability it contains a landmark $r$ on the detour part with $\dist(r, t) \le \Delta_k$.
          In addition, as $\dist(r, t) \le \Delta_k \le \dist(v, t)$, no shortest $r$--$t$ path traverses $e$, and hence the value of the state admits the \emph{landmark decomposition}
          $
              \dist_{G \setminus e}(s, t) = \dist_{G \setminus e}(s, r) + \dist(r, t)
          $
          (\cref{lem:far}), whose second term is already known from the BFS from $r$.
          The cost of this case therefore comes down to two tasks, computing $\dist_{G \setminus e}(s, r)$ and finding a landmark $r$ on the detour part.
          \textsf{GJM} provides the former in advance by a separate preprocessing using a generalization of the distance sensitivity oracle of Bernstein and Karger~\cite{BK09} to $\sigma$ sources, which is where the bulk of its work goes, and which relies on the undirectedness of the graph.
          The latter is handled by the multiscale sampling, the main idea of \cite{GJM20}, which limits the number of candidates to be scanned per state.

    \item \emph{Near states with a long replacement path}, i.e., those whose replacement path exceeds $\dist(s, u)$ by more than a constant times the smallest scale.
          Here the detour part is longer than the smallest scale even though $\dist(v, t)$ is not, so the landmarks of level $0$ still hit it and the same decomposition at a landmark $r$ is used.
          Its validity, however, can no longer be argued as above, since $\dist(r, t) \le \dist(v, t)$ is unavailable.
          \textsf{GJM} instead rules out the bad case that every shortest $r$--$t$ path traverses $e$, in which case such a path has the form $r, \dots, u, v, \dots, t$ and hence puts $u$ within distance $\dist(r, t) \le \Delta_0$ of $r$.
          Then the prefix of $\pi(s, t)$ up to $u$, the reverse of the $r$--$u$ subpath of that path, and the $r$--$t$ suffix of the replacement path form, in this order, an $s$--$t$ walk avoiding $e$ of length less than $\dist(s, u) + 2\Delta_0$, which is a replacement path shorter than assumed.
          The middle leg traverses a subpath of the $r$--$t$ path in the reverse direction, and is the other place where the undirectedness of the graph is used.
    \item \emph{Near states with a short replacement path.}
          These are evaluated by the trivial last-edge recurrence: a replacement path for $(t, e)$ ends with an edge $(y, t) \neq e$, and deleting that edge leaves a replacement path for the state $(y, e)$.
          \textsf{GJM} encodes these recurrences as an auxiliary graph whose vertices are the near states and evaluates all of them at once by a single run of Dijkstra's algorithm on it.
          Since this auxiliary graph contains the near states only, the recurrence closes only as long as the predecessor state $(y, e)$ is again near, and the shortness of the replacement path is precisely what guarantees this; this is why this case has to be separated from case~(2).
\end{enumerate}

\paragraph{Our idea.}
We keep the classification of the states by $\dist(v, t)$ and the multiscale landmark sampling of \textsf{GJM} as they are, and build an auxiliary graph in essentially the same way as \textsf{GJM} does for case~(3).
Our new ingredient is the observation that the quantity $\dist_{G \setminus e}(s, r)$ consumed by cases~(1) and~(2) is not external to the problem at all. More precisely, one can see the landmark decomposition as a recurrence formula.
Accordingly, the auxiliary graph $H_s$ we build for each source $s$ has as its vertices \emph{all} the states $(t, e)$, near and far alike, together with one base vertex $b_w$ per vertex $w \in V$ that carries the failure-free distance $\dist(s, w)$.
The landmark decomposition of case~(1) is then embedded into $H_s$ as edges entering the state vertices.
Namely, for a $k$-far state $(t, e)$ and each landmark $r$ of its own level with $\dist(r, t) \le \Delta_k$, we add an edge of weight $\dist(r, t)$ that enters the state vertex of $(t, e)$ and leaves either $b_r$ if $e \notin \pi(s, r)$, or $(r, e)$ otherwise (i.e., $e \in \pi(s, r)$).

Note that closing the recurrences on the full space of states, as explained above, removes the obstacles in cases~(1) and~(2).
Owing to this recursive structure, the construction requires no preprocessing of the source-to-landmark replacement distances.
In addition, case~(2) no longer needs to be distinguished.
The reason why \textsf{GJM} restricts the last-edge recurrence to short replacement paths is that its auxiliary graph has no vertex for a far state, so the recurrence closes only as long as the predecessor state $(y, e)$ is again near.
In $H_s$, by contrast, the predecessor always has a place to come from, namely the state vertex of $(y, e)$ if $e \in \pi(s, y)$, whether that state is near or far, and $b_y$ otherwise.
Hence the last-edge recurrence applies to every near state with no restriction on the length of the replacement path, and its analysis never uses the reverse traversal of paths seen in case~(2).
The undirectedness of the graph is used in \textsf{GJM} only here, in the validity argument of case~(2), and in the separate preprocessing of case~(1); since our approach bypasses both, the proposed algorithm works on directed graphs.

\section{Preliminaries}

\subsection{Notation}
The notation $G = (V, E)$, $n = |V|$, $m = |E|$, $S \subseteq V$, $\sigma = |S|$, and $G \setminus e$ follows the definitions given in \cref{subsec:result}, and the notation and terminology introduced in \cref{subsec:techniques} are made precise below.
The input graph $G$ is assumed to be simple, unweighted, and directed throughout the paper.
For a directed graph $H$, possibly with nonnegative edge weights, and a pair of vertices $u, v$ of $H$, we denote by $\dist_H(u, v)$ the shortest-path distance from $u$ to $v$ in $H$, where the length of a path is the number of its edges if $H$ is unweighted and the total weight of its edges otherwise.
We omit the subscript $H$ when $H = G$.
For a path $P$, we denote by $|P|$ the length of $P$.
For each source $s \in S$, we denote by $T_s$ a BFS tree rooted at $s$ in $G$, and by $\pi(s, t)$ the shortest $s$--$t$ path determined by $T_s$.

We use the term \emph{high probability} to mean with probability at least $1 - n^{-c}$ for any constant $c > 0$.
Throughout the paper, we fix an arbitrary such constant $c$; the parameters of our algorithm and the constant factors in our analysis depend on it.
For an integer $k \ge 0$, we write $\Delta_k = (c + 5) \cdot 2^k \sqrt{n/\sigma} \log n$.

\subsection{Multiscale Landmarks}
Let $\ell = \lfloor \log \sqrt{\sigma n} \rfloor$.
For each level $k = 0, 1, \dots, \ell$, we sample each vertex independently with probability $\frac{1}{2^k}\sqrt{\sigma/n}$,
and let $\mathcal{L}_k$ be the resulting landmark set of level $k$.
We also write $\mathcal{L} = \bigcup_{k=0}^{\ell} \mathcal{L}_k$.

\begin{lemma}\label{lem:landmark-size}
    With high probability, $|\mathcal{L}_k| = \tilde{O}(\sqrt{\sigma n}/2^k)$ holds for every $0 \le k \le \ell$ simultaneously.
\end{lemma}

\begin{proof}
    The expected size of $\mathcal{L}_k$ is $\mathbb{E}[|\mathcal{L}_k|] = n \cdot \frac{1}{2^k} \sqrt{\sigma/n} = \frac{\sqrt{\sigma n}}{2^k}$.
    Using Chernoff's bound, we obtain
    $
        \Pr\left[|\mathcal{L}_k| \ge (1+\delta)\,\mathbb{E}[|\mathcal{L}_k|]\right] \le e^{-\delta \mathbb{E}[|\mathcal{L}_k|]/3}
    $
    for any $\delta \ge 1$.
    Setting $\delta = 3(c+2) \log n$, we get
    \[
        \Pr\left[|\mathcal{L}_k| \ge \frac{(3c+7)\sqrt{\sigma n}\log n}{2^k}\right]
        \le e^{-\frac{(c+2)\sqrt{\sigma n}\log n}{2^k}}
        \le e^{-(c+2)\log n} \le n^{-(c+2)},
    \]
    where the second inequality uses $2^k \le 2^{\ell} \le \sqrt{\sigma n}$.
    Thus $|\mathcal{L}_k| = \tilde{O}(\sqrt{\sigma n}/ 2^k)$ holds with probability at least $1 - n^{-(c+2)}$.
    Since $\ell + 1 \le n$, a union bound over the $\ell + 1 = O(\log n)$ levels shows that the claim holds for all $0 \le k \le \ell$ simultaneously with probability at least $1 - n^{-(c+1)}$.
\end{proof}

\begin{corollary}
    $|\mathcal{L}| = \tilde{O}(\sqrt{\sigma n})$ holds with high probability.
\end{corollary}

As a preprocessing step, we run a BFS from each landmark $r \in \mathcal{L}$ and compute
$\dist(r, t)$ for all $t \in V$.
This takes $\tilde{O}(m\sqrt{\sigma n})$ time.

\subsection{Near State and $k$-Far State}
For a source $s \in S$, a terminal $t \in V$, and an edge $e = (u, v)$ on $\pi(s, t)$, we call the pair $(t, e)$ a \emph{state}.
We classify each state $(t, e)$ according to the distance $\dist(v, t)$ from the failed edge to the terminal as follows.

\begin{definition}[near state and $k$-far state]\label{def:near-far}
    A state $(t, e)$ with $e = (u, v)$ is
    \begin{itemize}
        \item \emph{near} if $\dist(v, t) < \Delta_0$, and
        \item \emph{$k$-far} ($0 \le k \le \ell$) if $\Delta_k \le \dist(v, t) < 2\Delta_k = \Delta_{k+1}$.
    \end{itemize}
\end{definition}
Note that every state that is not near is $k$-far for some $k$.
Our algorithm uses the following basic property mentioned in~\cite{GJM20} (see \cref{fig:far-lemma} for an illustration).

\begin{figure}[t]
    \centering
    \includegraphics[width=0.7\textwidth]{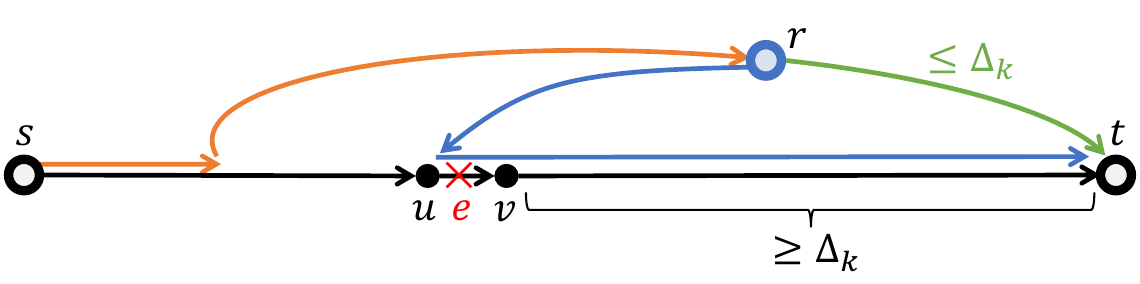}
    \caption{An illustration of \cref{lem:far}. For a $k$-far state $(t, e)$ with $e = (u, v)$, a shortest $s$--$t$ path in $G \setminus e$ decomposes at a landmark $r \in \mathcal{L}_k$ with $\dist(r, t) \le \Delta_k$ into a replacement path from $s$ to $r$ (orange) and a shortest $r$--$t$ path (green) with high probability. If a shortest $r$--$t$ path traversed $e$ (blue), its subpath from $v$ to $t$ would give $\dist(v, t) < \Delta_k$, contradicting $\dist(v, t) \ge \Delta_k$ from the $k$-far condition.}
    \label{fig:far-lemma}
\end{figure}

\begin{lemma}\label{lem:far}
    With high probability, for every source $s \in S$ and every $k$-far state $(t, e)$ ($0 \le k \le \ell$),
    \[
        \dist_{G \setminus e}(s, t)
        = \min_{\substack{r \in \mathcal{L}_k \setminus \{t\} \\ \dist(r, t) \le \Delta_k}}
        \left\{ \dist_{G \setminus e}(s, r) + \dist(r, t) \right\}.
    \]
\end{lemma}

\begin{proof}
    Fix a source $s \in S$ and a $k$-far state $(t, e)$ with $e = (u, v)$, and let $D$ denote the right-hand side of the claimed equality.

    We first show $\dist_{G \setminus e}(s, t) \le D$.
    Let $r \in \mathcal{L}_k \setminus \{t\}$ be any vertex with $\dist(r, t) \le \Delta_k$ (if no such $r$ exists, then $D = \infty$ and the inequality is trivial).
    Then, any shortest $r$--$t$ path in $G$ avoids $e$ because if it contained $e$, its subpath from $v$ to $t$ would give $\dist(v, t) \le \dist(r, t) - 1 < \Delta_k$, contradicting $\dist(v, t) \ge \Delta_k$ from the $k$-far condition.
    Hence $\dist_{G \setminus e}(r, t) = \dist(r, t)$, and the triangle inequality yields
    \[
        \dist_{G \setminus e}(s, t)
        \le \dist_{G \setminus e}(s, r) + \dist_{G \setminus e}(r, t)
        = \dist_{G \setminus e}(s, r) + \dist(r, t).
    \]
    Taking the minimum over all such $r$ gives $\dist_{G \setminus e}(s, t) \le D$.

    We next show that $D \le \dist_{G \setminus e}(s, t)$ holds with probability at least $1 - n^{-(c+4)}$.
    We may assume $\dist_{G \setminus e}(s, t) < \infty$, as the inequality is trivial otherwise.
    Fix a shortest $s$--$t$ path $P$ in $G \setminus e$.
    Note that $P$ is determined by $G$, $e$, $s$, and $t$, independently of the sampling of $\mathcal{L}_k$.
    Since $v$ lies on $\pi(s, t)$, we have
    $\dist_{G \setminus e}(s, t) \ge \dist(s, t) = \dist(s, v) + \dist(v, t) \ge \Delta_k$,
    so $P$ contains at least $\lfloor \Delta_k \rfloor$ vertices other than $t$.
    Let $W$ be the set of $\lfloor \Delta_k \rfloor$ vertices immediately preceding $t$ on $P$ (these are distinct, as $P$ is simple).
    Each vertex of $W$ belongs to $\mathcal{L}_k$ independently with probability $p_k = \frac{1}{2^k}\sqrt{\sigma/n}$, and since $\lfloor \Delta_k \rfloor > \Delta_k - 1$ and $p_k \le 1$, we have $p_k \lfloor \Delta_k \rfloor > p_k \Delta_k - 1 = (c+5) \log n - 1 \ge (c+4) \log n$, where the last inequality uses $\log n \ge 1$.
    Hence
    \[
        \Pr[W \cap \mathcal{L}_k = \emptyset]
        = (1 - p_k)^{\lfloor \Delta_k \rfloor}
        \le e^{-p_k \lfloor \Delta_k \rfloor}
        \le n^{-(c+4)}.
    \]
    Suppose $W \cap \mathcal{L}_k \neq \emptyset$, and take any $r \in W \cap \mathcal{L}_k$, which satisfies $r \neq t$.
    Let $P_1$ and $P_2$ be the subpaths of $P$ from $s$ to $r$ and from $r$ to $t$, respectively.
    Then $\dist(r, t) \le |P_2| \le \Delta_k$, so $r$ is a feasible candidate of the minimum defining $D$, and
    \[
        D \le \dist_{G \setminus e}(s, r) + \dist(r, t)
        \le |P_1| + |P_2|
        = \dist_{G \setminus e}(s, t).
    \]
    Therefore, for each fixed $s$ and $(t, e)$, the claimed equality fails with probability at most $n^{-(c+4)}$.
    Since the number of states over all sources is at most $\sigma n^2 \le n^3$, a union bound shows that the equality holds for all sources and all $k$-far states simultaneously with probability at least $1 - n^{-(c+1)}$.
\end{proof}

\section{Our Algorithm}

In this section, we present a very simple algorithm for solving the directed multiple source replacement paths problem in $\tilde{O}(m\sqrt{\sigma n} + \sigma n^2)$ time.

\paragraph{The Construction of Auxiliary Graphs.}\label{para:graph-construction}
For each source $s \in S$, we construct a weighted directed auxiliary graph $H_s$.
The vertex set of $H_s$ consists of
\begin{enumerate}
    \item a super source $z_s$,
    \item a base vertex $b_v$ for each $v \in V$, and
    \item a state vertex $q_{t, e}$ for each state $(t, e)$.
\end{enumerate}
The edge set of $H_s$ consists of the following three types of edges.
\begin{enumerate}
    \item Base edge:
          For each vertex $v \in V$, add an edge $(z_s, b_v)$ of weight $\dist(s, v)$.
    \item Near state edge:
          For each near state $q_{t, e}$ and each incoming edge $(v, t) \neq e$ of $t$,
          \begin{itemize}
              \item if $e \notin \pi(s, v)$, add an edge $(b_v, q_{t, e})$ of weight $1$, and
              \item if $e \in \pi(s, v)$, add an edge $(q_{v, e}, q_{t, e})$ of weight $1$.
          \end{itemize}
    \item Far state edge:
          For each $k$-far state $q_{t, e}$ and each $r \in \mathcal{L}_k \setminus \{t\}$ with $\dist(r, t) \le \Delta_k$,
          \begin{itemize}
              \item if $e \notin \pi(s, r)$, add an edge $(b_r, q_{t, e})$ of weight $\dist(r, t)$, and
              \item if $e \in \pi(s, r)$, add an edge $(q_{r, e}, q_{t, e})$ of weight $\dist(r, t)$.
          \end{itemize}
\end{enumerate}

We can show that the shortest-path distance from $z_s$ to $q_{t, e}$ in $H_s$ equals the shortest-path distance from $s$ to $t$ avoiding $e$ in $G$ for any state $(t, e)$.

\begin{lemma}\label{lem:correctness}
    With high probability, $\dist_{H_s}(z_s, q_{t, e}) = \dist_{G \setminus e}(s, t)$ holds for every source $s \in S$ and every state $(t, e)$ simultaneously.
\end{lemma}

\begin{proof}
    Throughout the proof, we write $e = (u, v)$.
    The direction ($\ge$) holds deterministically, and we prove the direction ($\le$) under the event that the equality of \cref{lem:far} holds for all sources and all $k$-far states, which occurs with probability at least $1 - n^{-(c+1)}$.

    ($\ge$) We show that every $z_s$--$q_{t, e}$ path $P$ in $H_s$ corresponds to an $s$--$t$ walk in $G \setminus e$ of the same length.
    By construction, every edge entering a state vertex of failed edge $e$ leaves either a base vertex $b_y$ with $e \notin \pi(s, y)$ or another state vertex $q_{y, e}$, so $P$ has the form $z_s, b_{u_0}, q_{u_1, e}, \dots, q_{u_j, e}$ with $u_j = t$.
    Each edge of $P$ is realized by a walk in $G \setminus e$ of the same length:
    the base edge $(z_s, b_{u_0})$ by the tree path $\pi(s, u_0)$, which avoids $e$;
    a near state edge of weight $1$ by the corresponding arc $(u_i, u_{i+1}) \neq e$;
    and a far state edge of weight $\dist(u_i, u_{i+1})$ by a shortest $u_i$--$u_{i+1}$ path.
    The last path avoids $e$: otherwise its subpath from $v$ would give $\dist(v, u_{i+1}) < \dist(u_i, u_{i+1}) \le \Delta_k$, contradicting $\dist(v, u_{i+1}) \ge \Delta_k$ from the $k$-far condition of $(u_{i+1}, e)$.
    Concatenating these walks yields an $s$--$t$ walk in $G \setminus e$ of the same length as $P$.

    ($\le$) We show $\dist_{H_s}(z_s, q_{t, e}) \le d$ for $d = \dist_{G \setminus e}(s, t)$, by strong induction on $d$ (the case $d = \infty$ is trivial).
    If $(t, e)$ is near, let $(y, t)$ be the last arc of a shortest $s$--$t$ path in $G \setminus e$ (note that $d \ge 1$, since $e \in \pi(s, t)$ implies $t \neq s$).
    Then $(y, t) \neq e$, and $\dist_{G \setminus e}(s, y) = d - 1$: the prefix of the path ending at $y$ gives $\dist_{G \setminus e}(s, y) \le d - 1$, while appending the arc $(y, t)$ to a shortest $s$--$y$ path in $G \setminus e$ gives $d \le \dist_{G \setminus e}(s, y) + 1$.
    We set $w = 1$.
    If $(t, e)$ is $k$-far, the conditioned event of \cref{lem:far} gives a landmark $y \in \mathcal{L}_k \setminus \{t\}$ with $\dist(y, t) \le \Delta_k$ and $\dist_{G \setminus e}(s, y) + \dist(y, t) = d$, and we set $w = \dist(y, t) \ge 1$.
    In both cases, $H_s$ contains a state edge of weight $w$ entering $q_{t, e}$ that leaves $b_y$ if $e \notin \pi(s, y)$, and $q_{y, e}$ if $e \in \pi(s, y)$.
    In the former case, $\dist_{H_s}(z_s, b_y) = \dist(s, y) = \dist_{G \setminus e}(s, y)$, since the only edge entering $b_y$ is the base edge and $\pi(s, y)$ survives in $G \setminus e$.
    In the latter case, $(y, e)$ is a state with $\dist_{G \setminus e}(s, y) = d - w < d$, so the induction hypothesis gives $\dist_{H_s}(z_s, q_{y, e}) \le \dist_{G \setminus e}(s, y)$.
    Either way, $\dist_{H_s}(z_s, q_{t, e}) \le \dist_{G \setminus e}(s, y) + w = d$.
\end{proof}

\paragraph*{Algorithm and Analysis.}
For each auxiliary graph $H_s$, we run Dijkstra's algorithm from $z_s$ and output
$\dist_{H_s}(z_s, q_{t, e})$ for each state $(t, e)$.
The details are given in \cref{alg:msrp}.

\begin{algorithm}[H]
    \caption{Directed MSRP}\label{alg:msrp}
    \KwIn{A directed graph $G = (V, E)$ and a source set $S \subseteq V$.}
    \KwOut{The distance $\dist_{G \setminus e}(s, t)$ for each source $s \in S$, terminal $t \in V$, and edge $e \in \pi(s,t)$.}
    \BlankLine
    \For{each source $s \in S$}{
        Construct the auxiliary graph $H_s$ by the procedure described \mbox{ in \hyperref[para:graph-construction]{``The Construction of Auxiliary Graphs''}}\;
        $\dist_{H_s}(z_s, \cdot) \gets$ \textbf{Dijkstra}$(H_s, z_s)$\;
        \lFor{each state $(t, e)$}{
            $\dist_{G \setminus e}(s, t) \gets \dist_{H_s}(z_s, q_{t, e})$
        }
    }
\end{algorithm}

\begin{lemma}\label{lem:edge-count}
    With high probability, the number of edges of $H_s$ is $\tilde{O}(m\sqrt{n/\sigma} + n^2)$ for every source $s \in S$ simultaneously.
\end{lemma}

\begin{proof}
    We condition on the event of \cref{lem:landmark-size}, which occurs with probability at least $1 - n^{-(c+1)}$; note that the landmarks are sampled only once and shared by all sources.
    We bound the number of edges of each type.
    \begin{enumerate}
        \item The number of base edges is $n$.
        \item For near state edges, since the number of near states $(t, \cdot)$ is at most $\Delta_0 + 1$ for each $t$, we have
              \[
                  \sum_{t \in V} \indeg(t) \times \#\{\text{near states } (t, \cdot)\}
                  \le m (\Delta_0 + 1)
                  = \tilde{O}(m \sqrt{n/\sigma}).
              \]
        \item For far state edges, since the number of $k$-far states is at most $\Delta_k + 1$ for each $t$, and $|\mathcal{L}_k| = \tilde{O}(\sqrt{\sigma n}/2^k)$ under the conditioned event, we have
              \[
                  \sum_{t \in V} \sum_{k=0}^{\ell} (\Delta_k + 1) |\mathcal{L}_k|
                  \le \sum_{t \in V} \sum_{k=0}^{\ell} \tilde{O}(n)
                  = \tilde{O}(n^2). \qedhere
              \]
    \end{enumerate}
\end{proof}

\begin{lemma}\label{lem:running-time}
    The running time of \cref{alg:msrp} is $\tilde{O}(m\sqrt{\sigma n} + \sigma n^2)$ with high probability.
\end{lemma}

\begin{proof}
    The preprocessing (BFS from the landmarks) takes $\tilde{O}(m\sqrt{\sigma n})$ time.

    Next, we show that each auxiliary graph $H_s$ can be constructed in $\tilde{O}(m\sqrt{n/\sigma} + n^2)$ time.
    We first run a BFS from $s$ to obtain the tree $T_s$ and the distances $\dist(s, v)$ for all $v \in V$ in $O(m)$ time.
    The states are then enumerated in $O(n^2)$ time by walking up $\pi(s, t)$ from each terminal $t \in V$, and each state $(t, e)$ with $e = (x, y)$ is classified as near or $k$-far in $O(1)$ time via $\dist(y, t) = \dist(s, t) - \dist(s, y)$.
    We store the enumerated states in an $n \times n$ table indexed by the pair $(t, y)$, which identifies the state $(t, e)$ since $e$ is the unique edge of $T_s$ entering $y$.
    This table allows us to decide whether $e \in \pi(s, v)$ for any vertex $v \in V$ in $O(1)$ time: by the definition of states, $e \in \pi(s, v)$ if and only if $(v, e)$ is a state, and in this case the lookup also returns the vertex $q_{v, e}$.
    The near state edges are generated by scanning the incoming edges of $t$ for each near state $(t, e)$, and the far state edges by scanning all landmarks $r \in \mathcal{L}_k$, together with the precomputed distances $\dist(r, t)$, for each $k$-far state $(t, e)$.
    Each scanned candidate is processed in $O(1)$ time using a table lookup, and the total number of scanned candidates is exactly the quantity bounded in the proof of \cref{lem:edge-count}.
    Hence, the construction of $H_s$ takes $\tilde{O}(m\sqrt{n/\sigma} + n^2)$ time.

    By \cref{lem:edge-count}, Dijkstra's algorithm on $H_s$ also runs in $\tilde{O}(m\sqrt{n/\sigma} + n^2)$ time for each $s \in S$.
    Hence, the total running time is
    $\tilde{O}(m\sqrt{\sigma n} + \sigma n^2).$
\end{proof}

\paragraph{The Proof of \cref{thm:main}.}
Finally, we prove our main theorem.

\begin{proof}[Proof of \cref{thm:main}]
    We show that \cref{alg:msrp} is such an algorithm.
    Its guarantees rest on two events: the event of \cref{lem:landmark-size}, which fails with probability at most $n^{-(c+1)}$ and under which \cref{alg:msrp} runs in $\tilde{O}(m\sqrt{\sigma n} + \sigma n^2)$ time (\cref{lem:edge-count,lem:running-time}); and the event of \cref{lem:far}, which fails with probability at most $n^{-(c+1)}$ and under which all the output distances are correct (\cref{lem:correctness}).
    By a union bound, both events hold simultaneously with probability at least $1 - 2n^{-(c+1)} \ge 1 - n^{-c}$.
    Since the constant $c > 0$ is arbitrary, the algorithm succeeds with high probability.
    Finally, the error of the algorithm is one-sided: the direction ($\ge$) of \cref{lem:correctness} holds deterministically, so every output value is at least the true replacement path distance, and the algorithm never underestimates it.
\end{proof}

\section{Concluding Remarks}
We presented a simple randomized combinatorial algorithm that solves the directed multiple source replacement paths problem in $\tilde{O}(m\sqrt{\sigma n} + \sigma n^2)$ time (\cref{thm:main}).
The first term matches the conditional lower bound of $m{(\sigma n)}^{1/2-o(1)}$ proved by Gupta et al.~\cite{GJM20} for combinatorial algorithms, and the second term is proportional to the output size.
Hence, our running time is near-optimal among combinatorial algorithms.
We conclude the paper with a few open questions:
\begin{itemize}
    \item \textbf{Derandomization.}
          Our algorithm is randomized since the multiscale landmarks are chosen by random sampling.
          For the RP and SSRP problems, deterministic algorithms matching the running times of their randomized counterparts are known~\cite{ACC19,BCFS21,BCCFS22}.
          Can these techniques, in particular the hierarchical derandomization framework of Bil{\`o} et al.~\cite{BCCFS22} for the directed SSRP problem, be leveraged to solve the directed MSRP problem deterministically in $\tilde{O}(m\sqrt{\sigma n} + \sigma n^2)$ time?
    \item \textbf{Multiple source distance sensitivity oracle.}
          A DSO can be significantly smaller than the explicit output of the corresponding replacement paths problem, and hence may be constructed faster than the problem can be solved explicitly.
          Indeed, in the single source setting, the oracle of Dey and Gupta~\cite{DG22} for undirected unweighted graphs has near-optimal size $\tilde{O}(n^{3/2})$ and $\tilde{O}(1)$ query time, and can be constructed in $\tilde{O}(m\sqrt{n})$ time, which matches the conditional lower bound for combinatorial SSRP algorithms and avoids the $\Theta(n^2)$ output size term inherent in solving the SSRP problem explicitly.
          In the multiple source setting, the oracle of Gupta and Singh~\cite{GS18} for undirected unweighted graphs answers queries from any of the $\sigma$ sources in $\tilde{O}(1)$ time using only $\tilde{O}(\sqrt{\sigma} n^{3/2})$ space, which is significantly smaller than the $O(\sigma n^2)$ words occupied by the explicit MSRP output.
          However, no construction faster than solving the MSRP problem is known.
          It would be interesting to investigate whether the techniques developed for the replacement paths problems, including ours, yield a multiple source DSO that can be constructed faster, ideally in $\tilde{O}(m\sqrt{\sigma n})$ time without the output size term $\sigma n^2$.
\end{itemize}

\paragraph{Acknowledgements.}
Kaito Harada was supported by JST BOOST, Japan Grant Number JPMJBS2402.
Taisuke Izumi was supported by JSPS KAKENHI Grant Number 23K24825.
This work was also supported by JST CRONOS Japan Grant Number JPMJCS24K2 and JST ASPIRE Japan Grant Number JPMJAP2302.

\paragraph{Use of AI tools.}
The main idea of our algorithm was obtained with the assistance of GPT-5.6 Pro (OpenAI).
The details of the algorithm and its analysis, including all proofs, were developed by the authors.
The authors also used Claude Fable 5 (Anthropic) to assist with the writing and editing of the manuscript.
All AI-assisted content was carefully reviewed and verified by the authors, who take full responsibility for the entire content of this paper.


\bibliographystyle{alphaurl}
\bibliography{references}

%

%

\end{document}